\documentclass[a4paper,onecolumn,11pt,accepted=2026-04-15]{quantumarticle}
\pdfoutput=1

\usepackage[utf8]{inputenc}
\usepackage[colorlinks=true,linkcolor=blue,citecolor=blue,urlcolor=blue]{hyperref}
\usepackage{graphicx}
\usepackage{amsmath,amssymb}
\usepackage{dsfont}
\usepackage{xcolor}
\usepackage{amsthm}
\usepackage{bbold}
\usepackage{float}
\usepackage{subfig}
\usepackage{braket}

\usepackage{algorithm}
\usepackage{algpseudocode}
\algnewcommand\algorithmicforeach{\textbf{for each}}
\algdef{S}[FOR]{ForEach}[1]{\algorithmicforeach\ #1\ \algorithmicdo}
\usepackage{tabularx}

\makeatletter
\newcommand{\multiline}[1]{%
  \begin{tabularx}{\dimexpr\linewidth-\ALG@thistlm}[t]{@{}X@{}}
    #1
  \end{tabularx}
}
\makeatother

\newtheorem{lemma}{Lemma}
\newtheorem{definition}{Definition}
\newtheorem{theorem}{Theorem}
\newtheorem{corollary}{Corollary}
\theoremstyle{remark}

\numberwithin{lemma}{section}
\numberwithin{definition}{section}
\numberwithin{theorem}{section}
\numberwithin{corollary}{section}
\numberwithin{remark}{section}

\numberwithin{algorithm}{section}

\newcommand{\oldcode}{\text{in}}
\newcommand{\newcode}{\text{wire}}
\newcommand{\gauge}{\mathcal{G}_{\text{wire}}}
\newcommand{\normalizer}[1]{\mathcal{N}(#1)}
\newcommand{\stabs}{\mathcal{S}_\oldcode}

\newcommand{\idty}{\mathbb{1}}
\newcommand{\oldcodespace}{\mathcal{C}_\oldcode}
\newcommand{\newcodespace}{\mathcal{C}_\newcode}
\newcommand{\logicalsnew}{\mathcal{L}_\newcode}
\newcommand{\logicalsold}{\mathcal{L}_\oldcode}
\newcommand{\barelogicals}{\mathcal{L}_\text{bare}}
\newcommand{\nnew}{n_\newcode}
\newcommand{\nold}{n_\oldcode}
\newcommand{\data}{\textsf{data}}
\newcommand{\anc}{\textsf{anc}}
\newcommand{\copyr}{\textsf{copy}}

\definecolor{FGreen}{RGB}{1,68,33}

\makeatletter
\def\l@subsection#1#2{}
\def\l@subsubsection#1#2{}
\makeatother

\begin{document}

\title{Wire codes}
\author{Nou\'edyn Baspin}
\affiliation{School of Physics, University of Sydney, Sydney, NSW 2006, Australia}
\orcid{0000-0002-4028-3098}
\author{Dominic J. Williamson}
\affiliation{School of Physics, University of Sydney, Sydney, NSW 2006, Australia}
\orcid{0000-0002-8029-6408}
\date{April 2026}

\begin{abstract}
\noindent
Quantum information is fragile and must be protected by a quantum error-correcting code for large-scale practical applications. 
Recently, highly efficient quantum codes have been discovered which require a high degree of spatial connectivity. 
This raises the question of how to realize these codes with minimal overhead under physical hardware connectivity constraints.  
Here, we introduce a general recipe to transform any quantum stabilizer code into a subsystem code that has local interactions, with weight and degree three, on a given graph. 
We call the subsystem codes produced by our recipe \textit{wire codes}, and their code parameters depend on the input code and the given graph. 
Wire codes can be adapted to have a local implementation on any graph that supports a low-density embedding of the input Tanner graph, with an overhead that depends on the embedding. 
In particular, applying our results to a stabilizer code and a subdivision of its own Tanner graph, yields a quantum weight reduction procedure with a multiplicative qubit overhead and distance reduction  that are linear in the input check degree and weight, respectively. 
Applying our results to hypercubic lattices leads to a construction of local subsystem codes with optimal scaling code parameters in any fixed spatial dimension. 
Similarly, applying our results to families of expanding graphs leads to local codes on these graphs with code parameters that depend on the degree of expansion. 
Our results constitute a general method to construct low-overhead subsystem codes on general graphs, which can be applied to adapt highly efficient quantum error correction procedures to hardware with restricted connectivity. 
\end{abstract}

\maketitle

\tableofcontents

\section{Introduction} 
\label{sec:Introduction}

Quantum error correction is a necessary ingredient for large-scale quantum technologies to function fault-tolerantly in the presence of noise~\cite{Shor1995,Steane1996,shor1996fault,gottesman1997stabilizer,aharonov1997fault,knill1998resilientQC,kitaev1997quantum,preskill1998reliable,Preskill1997fault}. 
Implementations of quantum error correction with the surface code have been the subject of a large volume of research due to their relative simplicity, low weight and degree checks, two-dimensional locality, and high thresholds~\cite{kitaev2003fault,bravyi1998quantum,dennis2002topological,raussendorf2007fault,Acharya2024}.
Subsystem codes have been discovered with even simpler implementations, lower weight and degree checks, and other favorable properties~\cite{Poulin2005,Bacon2005a,Bombin2009Interacting,Bombin2009,bravyi2011subsystem,Bravyi2013,ellison2022pauli,Brown2019Leakage}. 
Quantum low-density parity-check (qLDPC) codes~\cite{breuckmann2021ldpc} that achieve optimal asymptotic performance, i.e.~constant rate and relative distance, have been discovered recently~\cite{Panteleev2022}, see also Refs.~\cite{breuckmann2020balanced,leverrier2022quantum,Dinur2023}. 
This comes at the cost of simplicity in the sense that such codes have large constant weight and degree checks and require a much higher degree of connectivity than the surface code. 

The performance of subsystem codes that can be implemented locally on a lattice in Euclidean space with a fixed dimension $D$ are known to satisfy the bounds ${kd^{1/(D-1)}\leq O(n)}$, and ${d\leq O(n^{\frac{D-1}{D}})}$, where $n$ is the number of physical qubits, $k$ is the number of encoded logical qubits, and $d$ is the code distance
~\cite{bravyi2009no,bravyi2010tradeoffs,bravyi2011subsystem,flammia2017limits}. In particular, for codes that saturate the distance scaling bound we have
\begin{align}
\label{eq:SubsystemBPT}
[[n,k,d]]=[[n,O(n^{\frac{D-1}{D}}),O(n^{\frac{D-1}{D}})]]. 
\end{align} 
In this work, we introduce a construction of subsystem codes from an arbitrary input stabilizer code that reduce the weight and degree to three, and which can be adapted to fit the locality of any given connectivity graph. 
When we apply our construction to map good qLDPC codes to local implementations on hypercubic lattices and find a family of codes that saturate the bound in Eq.~\eqref{eq:SubsystemBPT} for any dimension. 

\subsection{Overview of Main Results} 

Here, we introduce a weight reduction procedure that maps any stabilizer code to a 3-local subsystem code with related code parameters.  
We leverage the flexibility of our weight reduction procedure to find local implementations of the weight reduced subsystem codes, given an embedding of the input Tanner graph into an arbitrary target graph. 
We call the resulting family of codes \textit{wire codes}. 
The code parameters of the input code are mapped to those of the output code as follows 
\begin{align}
    [[n,k,d]] \mapsto [[O( \ell_{\text{max}} \delta n),k,\Omega(\frac{1}{w}d)]],
\end{align}
where $\omega$ and $\delta$ are the maximum weight and degree of the input code, respectively, and $\ell_{\text{max}}$ is the length of the longest edge in the wire code. 
The weight of a check is the number of qubits in its support, while the degree of a qubit is the number of checks it interacts with. 
In the wire codes, there is an additional length parameter assigned to each edge in the input Tanner graph, which corresponds to how far the edge can reach in an embedding of that Tanner graph. 

Each wire code has a particular structure that is based on a trivalent resolution of the input Tanner graph, together with an assignment of variable lengths to the edges. 
The associated wire code has local checks along the edges, or wires, of the resolved Tanner graph. 
When all edge lengths are fixed to $1$, the wire code construction defines a weight reduction procedure that is described in the following theorem. 
\begin{theorem}[Wire code weight reduction]
    Applying the procedure described in Section~\ref{sec:WeightReduction} to an $[[n,k,d]]$ stabilizer code, with max check weight $\omega$ and qubit degree $\delta$, produces an $[[O( \delta n),k,\Omega(\frac{1}{\omega}d)]]$ subsystem code, which we call a wire code, with max check weight and degree three. 
\end{theorem}
In the above theorem statement we do not count single qubit gauge checks in the qubit degree. 
This is because such checks do not require additional ancillary qubits and interactions to implement. 

By incorporating general edge lengths, the wire code construction can be combined with a constant density embedding of the resolved Tanner graph into an arbitrary graph $G$ to find a local implementation of the input code on $G$. 
This result is summarized by the following theorem.
\begin{theorem}[Wire code graph embedding]
    Given an $[[n,k,d]]$ stabilizer code $\oldcodespace$ and a $c$-to-one embedding of the trivalent resolution of its Tanner graph into a graph $G$ the wire code described in Section~\ref{sec:GraphEmbedding} provides an $[[O(\ell_{\text{max}}\delta n),k,\Omega(\frac{1}{\omega}d)]]$ subsystem code that we say implements $\oldcodespace$ with local checks on $G$. 
    Here $\ell_{\text{max}}$ is the longest edge length assigned by the graph embedding. 
\end{theorem}

\begin{figure}[t]
    \centering
    \includegraphics[page=42, scale=1]{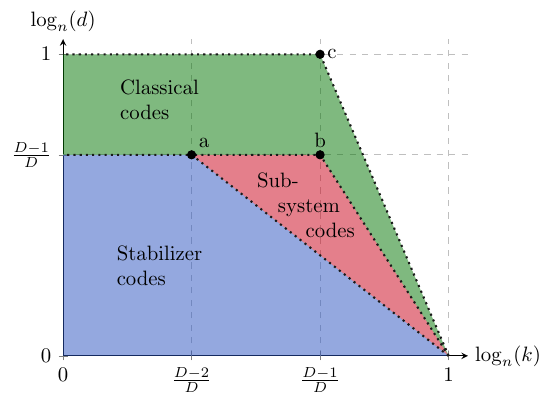}
    \caption{An illustration of existing bounds on code parameters for stabilizer, subsystem, and classical codes. 
    The stabilizer code region (blue) satisfies $kd^{2/(D-1)}\leq O(n)$. 
    The subsystem code region (red) satisfies $kd^{1/(D-1)}\leq O(n)$.
    The classical code region (green) satisfies $kd^{1/D} \leq O(n)$. 
    The point `b' corresponds to families of $D$-dimensional wire codes, see also the construction in Refs.~\cite{bravyi2011subsystem,bacon2015sparse}. 
    The points `a', and `c' correspond to the constructions in Refs.~\cite{portnoy2023local,Lin2023,williamson2024layercodes,Li2024}, and Ref.~\cite{baspin2023combinatorial}, respectively. }
\label{fig:CodeParams}
\end{figure}

Wire codes are holographic\footnote{Our construction is distinct from the holographic encoding isometries introduced in Ref.~\cite{Pastawski2015}.} in the sense that they can be understood as realizing a ``boundary'' code via local checks along chains of ancilla qubits that are routed through a ``bulk''. 
We make this intuition concrete by constructing families of wire codes that have local gauge checks in any given dimension $D\geq 2$. 
These $D$-dimensional wire codes implement the checks of an input stabilizer code on a $(D-1)$-dimensional boundary via ancillary qubits that route the connectivity of the checks through the $D$-dimensional bulk. 
When this construction is applied to a good qLDPC code~\cite{panteleev2020quantum, leverrier2022quantum, Dinur2023} the output $D$-dimensional wire code is optimal in the sense that it has code parameters that saturate the bounds of Ref.~\cite{bravyi2011subsystem} shown in Eq.~\eqref{eq:SubsystemBPT}. 
The parameters achieved by these families of codes are depicted in Fig.~\ref{fig:CodeParams} by the dotted line that joins the point `b' to (1,0).  
This result is captured by the corollaries below. 

\vspace{1cm} 

\begin{corollary}[Optimal local subsystem codes in $D$ dimensions]
\label{cor:SubsystemBPTSaturation}
    Wire codes constructed from $[[n,\Theta(n),\Theta(n)]]$ good qLDPC codes~\cite{Panteleev2022} and the graph embeddings into $D$ dimensional Euclidean space from Ref.~\cite{thompson1977sorting} have code parameters $[[n^{D/(D-1)},\Theta(n),\Theta(n)]]$ which saturate the bounds from Ref.~\cite{bravyi2011subsystem} for all $D\geq 2$. 
\end{corollary}
\begin{corollary}[Multiple blocks of optimal local subsystem codes in $D$ dimensions]
    $\Theta(N^{1-\alpha})$ blocks of size $\Theta(N^\alpha)$ containing optimal wire codes from Corollary~\ref{cor:SubsystemBPTSaturation} results in subsystem codes with parameters $[[\Theta(N),\Theta(N^{1-\alpha/D}),\Theta(N^{\alpha(D-1)/D)}]]$ saturating the bounds from Ref.~\cite{bravyi2011subsystem} for all $D\geq 2$ and $\alpha\in[0,1]$. 
\end{corollary}

While the literature on local codes in Euclidean space abounds, there are relatively few works on codes with non-Euclidean locality. 
This has quickly become a topic of interest as recent hardware developments have raised the possibility of architectures with connectivity that is neither local Euclidean, nor all-to-all~\cite{Moses2023,Evered2023,Storz2023,Alexander2024}. 
In an effort to formulate a generalised notion of local connectivity, some works have made use of graph theoretic ideas~\cite{baspin2021connectivity, baspin2023improved, baspin2021quantifying, baspin2023combinatorial}. Among these results, the \emph{expansion} of a graph has emerged as a useful notion of abstract connectivity~\cite{baspin2023combinatorial}. 
A corollary of our work is that the more expanding an architecture is, the better the codes are that we can provably embed into it while preserving locality. Here, we rely on the notion of edge expansion specified in Definition~\ref{def:exp}. 

\begin{corollary}[Local subsystem codes on graphs with expansion]
    Wire codes constructed from $[[n,\Theta(n),\Theta(n)]]$ good qLDPC codes~\cite{Panteleev2022, leverrier2022quantum} and graph embeddings into an infinite family $\{G_n\}_n$ of $\alpha$-expander graphs, yields a family $\{\mathcal{C}_n\}_n$ of subsystem codes that are local on $G_n$ with parameters $[[n, \Omega\left({\alpha n}/{\text{polylog}(n)}\right), \Omega\left({\alpha n}/{\text{polylog}(n)}\right)]]$.
\end{corollary}

Wire codes present a new method to address the challenge of implementing highly efficient codes, with highly connected Tanner graphs, under general connectivity constraints that arise from hardware considerations.  
The code properties of wire codes are summarized in Table~\ref{tab:summary}, and the construction of local wire codes in $D$-dimensions is summarized in Alogrithm~\ref{alg:high-d-wire-codes-alg}.

\setlength\tabcolsep{1.5cm}
\begin{table}[h]
    \centering
    \begin{tabular}{ c c}
    Input code                  &   Output wire code \\
    \hline \hline
    General stabilizer code      &   local subsystem code \\
    Physical qubit              &   Copy qubit branches \\
    Check                   &   Ancillary qubit branches \\
    Max check weight $w$       &   Max check weight 3 \\
    Max qubit degree $\delta$       &   Max qubit degree 3 \\
    $[[n,k,d]]$                 &   $[[O( \ell_{\text{max}} \delta n),k,\Omega(\frac{1}{w}d)]]$
\end{tabular}
    \caption{Summary of the input and output code properties in the wire code construction.}
    \label{tab:summary}
\end{table}

\subsection{Prior Work}

This work extends the literature on subsystem codes that are not derived from topological codes. Among the touchstone works on this subject, in Ref.~\cite{bravyi2011subsystem} Bravyi shows how to obtain optimal 2D subsystem codes, see Ref.~\cite{Devakul2020b} for an extension of this construction to 3D. Bacon et.~al.~\cite{bacon2015sparse} introduce a systematic method to map a circuit to a subsystem code local in $D$-dimensions. Their construction can be updated to saturate the bounds of Ref.~\cite{bravyi2011subsystem} in all dimensions given a fault-tolerant syndrome extraction circuit for a good quantum LDPC code~\cite{leverrier2022quantum,panteleev2020quantum,Dinur2023}.

Since the foundational work in Refs.~\cite{bravyi2011subsystem,bacon2015sparse}, obtaining local sparse codes from non-local structures has become a field of study in itself. 
Some examples based on stabilizer codes are Ref.~\cite{portnoy2023local} and Ref.~\cite{williamson2024layercodes} -- the construction of Ref.~\cite{portnoy2023local} was later made more explicit in Ref.~\cite{Lin2023} for balanced product codes~\cite{breuckmann2020balanced} and further improved in Ref.~\cite{Li2024}.
Other work has focused on weight reduction of stabilizer codes, including  Refs.~\cite{hastings2016weight,hastings2023quantum,Sabo2024}. 
Our weight reduction results achieve lower weight and degree than previous works by following a strategy that is reminiscent of the classical weight reduction procedure reviewed in Ref.~\cite{Sabo2024} using subsystem codes in place of stabilizer codes. 
Similarly, optimal local classical codes are obtained in Ref.~\cite{baspin2023combinatorial}, while non-Pauli Hamiltonians are explored in Ref.~\cite{apel2023simulatingldpccodehamiltonians} through the lens of Hamiltonian simulation.

\subsection{Open Questions} 

Our results raise a number of questions for future work. 

\textbf{Fault tolerance}: The wire codes that obtain optimal parameter scaling on finite dimensional Euclidean lattices have only global stabilizers. 
It is generally difficult to make subsystem codes with local checks and global stabilizers fault tolerant.\footnote{We remark that the wire code construction is expected to preserve fault-tolerance when only a constant number of additional qubits per original check and qubit are introduced, the overhead is expected to scale with the number of additional checks and qubits.} 
What additional procedures can be performed to make these codes fault tolerant, and what is their overhead? 
An interesting result in this direction for the Bacon-Shor code was presented in Ref.~\cite{Gidney2023}. 

\textbf{Hamiltonians}: $k$-local subsystem codes naturally define a family of Hamiltonians on the Tanner graph. 
These Hamiltonians generally have tunable parameters that can drive nontrivial phase transitions~\cite{ellison2022pauli} and interesting quantum dynamics~\cite{Wildeboer2021}. 
What are the phase diagrams of these Hamiltonians, and is the model $H = - \sum_g g$ gapped? The question of mapping a geometrically non-local system to a \emph{gapped} 2D local system was addressed in \cite{apel2023simulatingldpccodehamiltonians} albeit using radically different tools that circumvent the difficulties posed by subsystem Hamiltonians.

\textbf{Optimal codes on arbitrary graphs}: Can the wire code construction produce subsystem codes with optimal scaling code parameters given the connectivity of some more general families of graphs? Is the scaling achieved in Corollary \ref{cor:expander-wire-codes} optimal? 
We remark that for Tanner graphs with general connectivity, tight bounds on code parameters are yet to be formulated. See Appendix \ref{sec:upper-bound} for a partial result in this direction. 
Can Theorem \ref{thm:expander-wire-codes} be extended to map an arbitrary stabilizer code to a local \emph{stabilizer} code on an arbitrary graph?

\textbf{Floquet code weight reduction}: The subsytem codes we have introduced here defined a simple floquet code procedure with weight three checks. 
More generally, floquet code syndrome extraction procedures can always be found that have weight two checks. 
What is the most efficient such procedure given desired stabilizer code checks?

\subsection{Examples} 

We now illustrate our construction by depicting the 2D wire codes that are output for several simple examples of input codes. 
These examples are presented purely for the purposes of illustration as they are too small to benefit from our weight reduction procedure. 
For an introduction to the ideas used to construct these codes, see Section~\ref{sec:MainIdeas}. 

\subsubsection{Repetition Code}

The first example is based on the $$[[3,1,1]]$$ three qubit repitition code, which has stabilizer checks $ZZI,IZZ$. 
The Tanner graph of the $[[7,1,1]]$ wire code representation is shown below. 
\begin{equation}
    \label{fig:algstep0}
    \vcenter{\hbox{\includegraphics[page=16, scale=1]{diagrams-wire-codes.pdf}}}
\end{equation}
Here circles represent qubit nodes and boxes represent check nodes. 
By convention, the internal check nodes are not explicitly pictured between the ancillary qubits. 
There is an additional single body $X$ check for each ancillary qubit that is not explicitly depicted. 

\subsubsection{Five qubit code}

The $[[5,1,3]]$ error correcting code has the following stabilizer checks 
\begin{align}
    XZZXI,IXZZX,XIXZZ,ZXIXZ.
\end{align} 
The Tanner graph of the $[[47,1,\leq 3]]$ 2D wire code representation of this code is shown below. 
\begin{equation}
\label{fig:algstep0b}
    \vcenter{\hbox{
    \includegraphics[page=34, scale=.8]{diagrams-wire-codes.pdf}}}
\end{equation}
Here the dashed lines denote $X$-type Pauli operators involved in a check. 
Again there is an additional single body $X$ check for each ancillary qubit that is not depicted. 

\subsubsection{Shor's Code}

Shor's $[[9,1,3]]$ code has the following stabilizer checks~\cite{Shor1995}
\begin{align}
    \begin{array}{c c c c c c c c c }
        X & X & X & X & X & X & I & I & I  \\
        I & I & I & X & X & X & X & X & X  \\
        Z & Z & I & I & I & I & I & I & I  \\
        I & Z & Z & I & I & I & I & I & I  \\
        I & I & I & Z & Z & I & I & I & I  \\
        I & I & I & I & Z & Z & I & I & I  \\
        I & I & I & I & I & I & Z & Z & I  \\
        I & I & I & I & I & I & I & Z & Z .  
    \end{array}
\end{align}
The $[[135,1,\leq 3]]$ 2D wire code representation of Shor's code is shown below. 
\begin{equation}
\label{fig:algstep0c}
    \vcenter{\hbox{
    \includegraphics[page=11, scale=.75]{diagrams-wire-codes.pdf}}}
\end{equation}
The conventions in the figure above follow the previous examples.

\subsection{Section Outline} 

The remainder of this work is laid out as follows. 
In Section~\ref{sec:WireCodeConstruction} we introduce the wire code construction first informally and then formally. 
In Section~\ref{sec:CodeProperties} we provide proofs of the key properties of wire codes. 
In Section~\ref{sec:Discussion} we discuss our results.

\section{Wire Code Construction} 
\label{sec:WireCodeConstruction}

In this section we introduce the wire code construction for general graphs and explain how it leads to weight and degree reduction.

\subsection{Main ideas}
\label{sec:MainIdeas}

In this section we present an informal introduction to the main ideas we use to map a stabilizer code to a subsystem code. 
We demonstrate these ideas through a number of simple examples. 
We first introduce a graphical notation for a check in a subsystem code. A blue square represents a gauge check, and orange circles are qubits; the solid lines correspond to a $Z$ operators, while the dashed lines correspond to $X$ operators. Correspondingly, two parallel solid and dashed lines correspond to a $Y$ operator. In the following example, the check represented is $Z_1 X_2 Y_3 \dots X_t$
\\

\begin{equation}
    \vcenter{\hbox{\includegraphics[page=1]{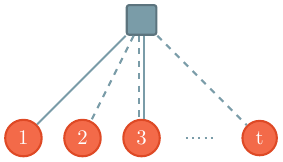}}}
\end{equation}
\\

\noindent Using the above notation, the 5-qubit code is represented as follows.
\\
\begin{equation}
    \vcenter{\hbox{
    \includegraphics[page=2]{diagrams-wire-codes.pdf}}}
\end{equation}
\\

\noindent Our construction aims to address the following problem: in some architecture, it might happen that the qubits involved in the check, say $s_1$, are far apart. It would then be beneficial to be able to break down this interaction into local terms. 

\paragraph{Breaking down $X_1Z_2Z_3X_4$, a first attempt.} 
Our first attempt focuses on breaking down the $X_1Z_2Z_3X_4$ stabilizer into local gauge checks. 
A naive approach, based on the strategy that works for classical codes~\cite{baspin2023combinatorial,Sabo2024}, is shown below. We break down $s_1$ into local checks and additionl ancilla qubits. 
\\
\begin{equation}
    \vcenter{\hbox{\includegraphics[page=3]{diagrams-wire-codes.pdf}}}
    \label{fig:WeightReduction}
\end{equation}
\\

\noindent We differentiate between the $\data$ register, i.e. the qubit from the original code, and the $\anc$ register, which contains the newly introduced qubits.

One can ascertain that for any normalizer of this new code, its restriction to the $\data$ subsystem has to be in the normalizer of the original code. Formally this emerges from the observation that multiplying the gauge checks together gives $X_1Z_2Z_3X_4\otimes\idty_\anc$. 
This is shown in the figure below. 
\\

\begin{equation}
    \vcenter{\hbox{\includegraphics[page=25]{diagrams-wire-codes.pdf}}}
\end{equation}
\\

\noindent In green we highlighted the gauge checks producted together, and the resulting operator can be seen to be $X_1Z_2Z_3X_4\otimes \idty_\anc$.
However, this does not prevent the introduction of small errors on the $\anc$ subsystem even when the $\data$ is merely acted on by $\idty_\data$. 
This is due to the fact that $\idty_\data$ is in the normalizer of $X_1Z_2Z_3X_4$. A $Z$ operator on the $\anc$ subsystem is now a logical:
\\
\begin{equation}
    \vcenter{\hbox{\includegraphics[page=8,scale=1]{diagrams-wire-codes.pdf}}}
    \label{fig:example1}
\end{equation}
\\

\paragraph{Removing small errors on the ancilla subsystem.} To remediate the issue of the single-site $Z$ logical operators, it suffices to add a single site $X$ check to each ancilla qubit, see below. 
\\
\begin{equation}
    \vcenter{\hbox{\includegraphics[page=4,scale=1]{diagrams-wire-codes.pdf} }}
    \label{fig:example2}
\end{equation}
\\
These single site $X$ operators ensure the elimination of the small errors on the ancilla subsystem, and are the origin of the non-commuting nature of the new checks.
Now, the number of logical qubits $k$ is preserved. 
Any logical of the original code can be mapped to a bare logical of the new subsystem code via multiplication with single site $X$ operators on the ancilla subsystem, such that the (anti-)commutation relations are preserved. 

\paragraph{A \textit{topological} mapping.} 
The construction we have outlined thus far can be used to break up checks in a stabilizer code into low-weight, local checks in a way that preserves $k$ and $d$. Importantly, the properties of this mapping do not depend on the precise geometric parameters of the new checks. This observation allows us to `stretch' the original stabilizers to any desired length, see Figure \ref{fig:WeightReductionExtension}. 

\begin{figure}[t]
    \centering
    \includegraphics[page=5, scale =.7]{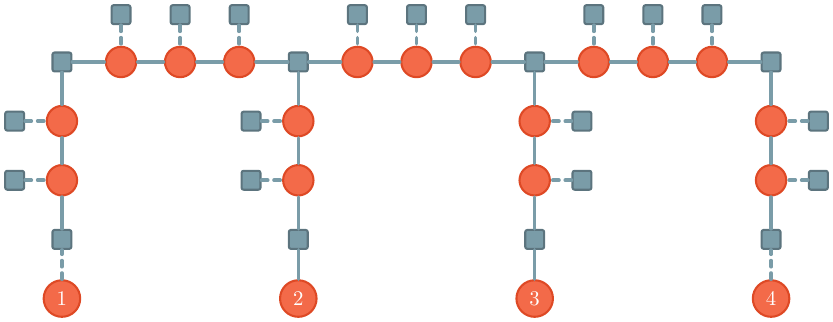}
    \caption{Showcasing how to implement the $X_1Z_2Z_3X_4$ stabilizer with local interactions when the data qubits are even further apart. To lighten the notation, two qubits joined by a solid line represent a $ZZ$ check between them.}
    \label{fig:WeightReductionExtension}
\end{figure}

It is now possible to give a comprehensive picture of the result. Once the $X_1Z_2Z_3X_4$ has been embedded through local checks, the resulting 5-qubit code is depicted in Figure~\ref{fig:xzzx-stretched}.
\begin{figure}[t]
    \centering
    \includegraphics[page=10, scale =.7]{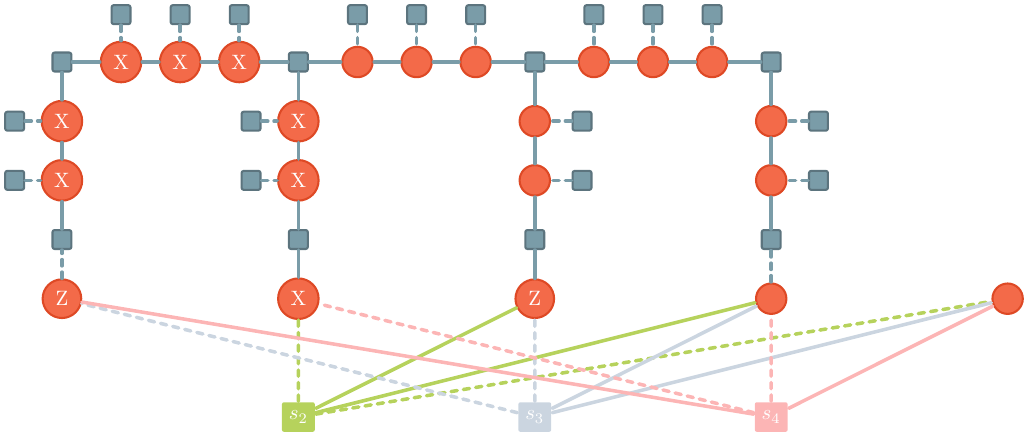}
    \caption{The 5-qubit code with the $X_1Z_2Z_3X_4$ stabilizer mapped to local checks. The Pauli operators are an example of a bare logical operator.}
    \label{fig:xzzx-stretched}
\end{figure}

\paragraph{Generalizing the mapping. }
Naturally, one could ask how to localize a general $4$-body stabilizer of the form $P_1P_2P_3P_4$ instead of $X_1Z_2Z_3X_4$. 
In doing so it is important to keep in mind the main reasons why the previous transformation can be expected to almost preserve the distance of the original code is because 1) by multiplying a subset of the gauge checks we can recover the original stabilizer on the data qubits, 2) we eliminate phase errors with single-site $X$ gauge checks. 
Recovering these same properties for $P_1P_2P_3P_4$ is straightforward: we replace the $Z\otimes Z$ checks acting on both the $\data$ and $\anc$ subsystems by checks of the form $P_i \otimes Z$, as illustrated in Figure \ref{fig:example3}. 

\begin{figure}[H]
    \centering
    \subfloat[\centering $P_1P_2P_3P_4$ stabilizer check]{{\includegraphics[page=22,scale=.8]{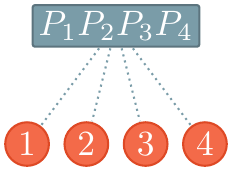} }}
    \qquad
    \subfloat[\centering  Code with local gauge checks emulating the $P_1P_2P_3P_4$ stabilizer. Some ancilla qubits are labelled $c_1, \dots, c_4$ for notational convenience.]{{\includegraphics[page=23,scale=.8]{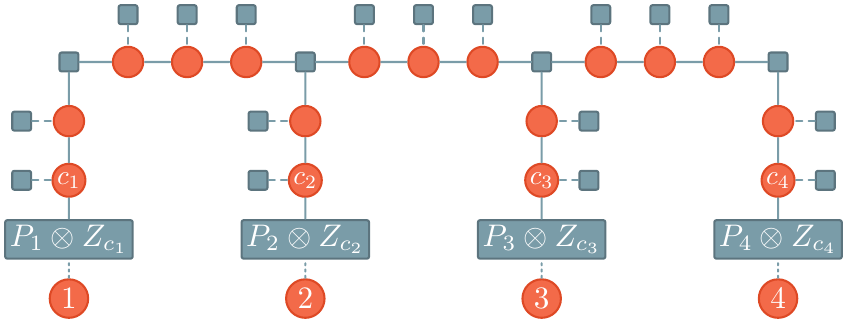}}}
    \caption{Dotted lines are meant to distinguish general Pauli operators from $Z$ (solid lines), $X$ (dashed lines) and $Y$ (dashed and solid double lines) operators.}
    \label{fig:example3}
\end{figure}

\paragraph{Degree reduction.} 
Along with a method to break down long-range checks, we also introduce a construction to obtain a code of degree at most three by adding ancilla qubits and gauge checks. This does not count single body gauge checks in the qubit degree. In the following figure we illustrate this idea by focusing on a qubit connected to multiple checks in the original code. On the left are the original checks, on the left is there degree-reduced version:
\begin{equation}
    \vcenter{\hbox{\includegraphics[page=26,scale=.8]{diagrams-wire-codes.pdf} }} \qquad \longrightarrow \qquad \vcenter{\hbox{\includegraphics[page=27,scale=.55]{diagrams-wire-codes.pdf}}}
    \label{eq:DegRed}
\end{equation}
The construction follows the same idea as the previous section: after adding ancillas and new gauge checks, we must be able to recover the original stabilizers, see Figure \ref{fig:degree-reduction-explanation}.

\begin{figure}[H]
    \centering
    \includegraphics[page=28, scale =.45]{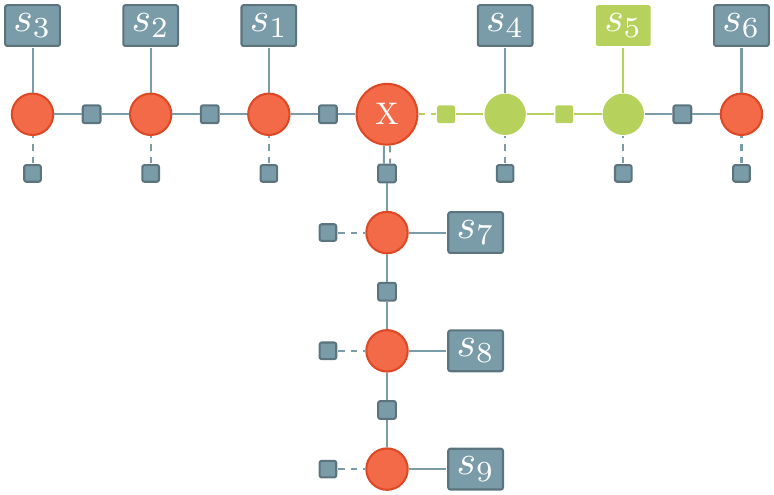}
    \caption{ The value of, for example, the stabilizer $s_5$ can be recovered by multiplying the value of the checks highlighted in green. Due to the nature of the checks, whenever the operator acting on qubit $1$ anticommutes with $X$, then the ancilla qubits connected to the checks $s_4, s_5, s_6$ have to be acted on by the operator $X$, in order to be in the normalizer. In a way we are `carrying' the commutator of the operator on qubit $1$ to $s_4, s_5$ and $s_6$.}
    \label{fig:degree-reduction-explanation}
\end{figure}

\paragraph{What is the dressed distance?}
The careful reader might notice that the weight of the dressed logical operators can be reduced by cleaning a bare logical onto the ancillary subsystem induced by a check. However, the single-site $X$-type checks cannot affect the operators on the data subsystem. Additionally, the $Z$-type checks can only either reduce to the application of a stabilizer of the original code on the data subsystem, or will \emph{always} leave a non-trivial operator on the ancilla subsystem. 
From these observations, one can convince themselves that the dressed distance is at least $\frac{1}{\omega}d_\oldcode$, see Lemma~\ref{lem:d} for further details.

\paragraph{Bringing everything together. }
Below is an illustration of the 2D wire code mapping applied to the $[[5,1,3]]$ code. To avoid clutter, we abstain from fully specifying the checks when they are clear from context, including the single-site $X$'s.
Here we have chosen generating checks $XZZXI,IXZZX,XIXZZ,ZXIXZ.$
\begin{equation}
    \vcenter{\hbox{\includegraphics[page=34, scale=.8]{diagrams-wire-codes.pdf}}}
\end{equation}

\subsection{Wire code construction in two dimensions}

In this section we formally state how our construction can be applied to map an arbitrary LDPC code to a 2D-local subsystem code. 

\paragraph{Degree reduction.}
The process of degree reduction can be succinctly described in 2D by introducing the notion of ``branches". A qubit in the support of $\delta_Z$ $Z$-operators, $\delta_X$ $X$-operators, and $\delta_Y$ $Y$-operators will be mapped to a $Z$-branch of length $\delta_Z$, an $X$-branch of length $\delta_X$, and a $Y$-branch of length $O(\delta_Y + \delta_X)$. 

\begin{definition}[$P$-branch]
\label{def:branch}
    For a target qubit `$q$' a $P$-branch of length $\delta$ consists of a ``copy" set of qubits $\{r_1, r_2, \dots, r_\delta\}$, along with the checks $P_qZ_{r_1}$, $Z_{r_i} Z_{r_{i+1}}$ for $i \in [1,\delta-1]$, and $X_{r_i}$ for $i \in [1,\delta]$.
\end{definition}

\begin{figure}[H]
    \centering
    \includegraphics[page=35,scale=.6]{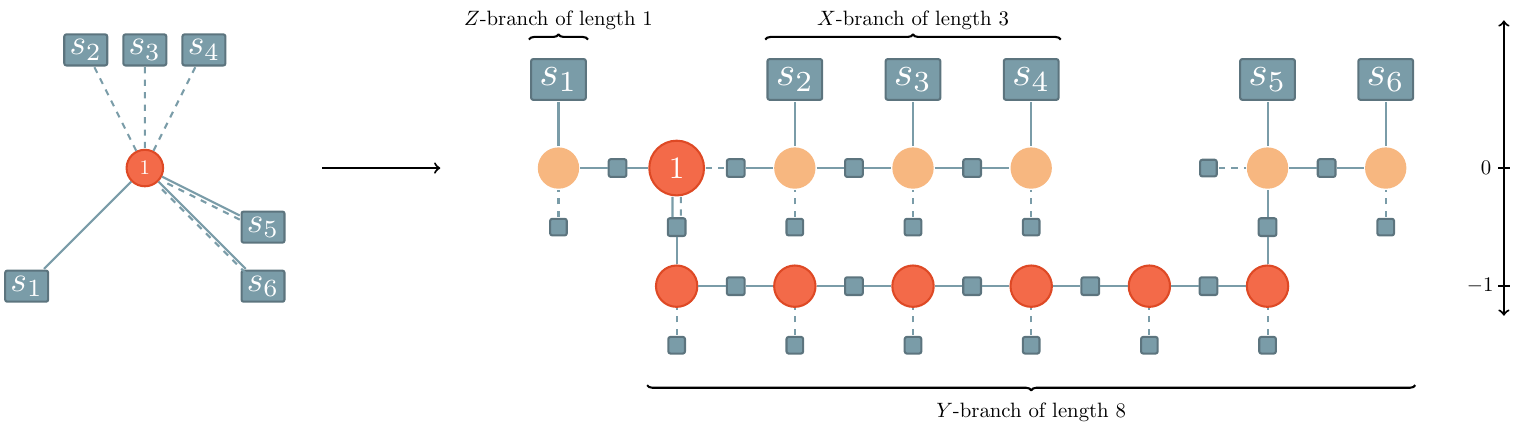} 
    \caption{Degree reduction in 2D. Different branches of different sizes. Note the $Y$-branch incurs an overhead of $\delta_X + \delta_Y$ due to the branch going `behind' the $X$-branch. This overhead is necessary in our layout, if we want to avoid stacking multiple qubits at a single location, due to the way the checks connect to the $Y$-branch.}
    \label{fig:example4}
\end{figure}

Reducing the degree of a qubit occupies a line of width $O(\delta_X+\delta_Z+\delta_Y)$. 
This is because, if we are to avoid stacking multiple qubits on the same site, a $Y$-branch of length $\geq 1$ has to extend past the $X$-branch in order to preserve the locality of the code. 
The peach coloured qubits in Fig.~\ref{fig:example4} are called \emph{copy} qubits and serve to implement a localized version of the checks.

Below we formalize the degree reduction procedure in 2D via Algorithm~\ref{alg:degree-reduction}. 

\begin{algorithm}[H]
\caption{DegreeReduce}
\label{alg:degree-reduction}
\begin{algorithmic}
    \Require qubit $q \in [1, \dots, n]$ of $\oldcodespace$
        \State Let $\delta_Z$ be the number of $Z$ operators $q$ is in the support of. Define $\delta_X,\delta_Y$ similarly.
        \State Introduce a $Z$-branch of length $\delta_Z$, and similarly for $X$ and $Y$. These branches will be contained between the coordinates $((q-1)\cdot\delta,0)$, and $(q\cdot\delta,0)$.
        \State For each Pauli $P\in \{X,Y,Z\}$, let $\{s_{1}, \dots, s_{\delta_P}\}$ be the set of $P$-type stabilizers supported on $q$. Then each of these stabilizers is now supported on a different copy qubit of $q$'s $P$-branch.
\end{algorithmic}
\end{algorithm}

\paragraph{Embedding stabilizers.}

To weight reduce the checks we map them to $Z$-branches that are connected to copy qubits of the appropriate Pauli type for the corresponding data qubits. 
For example, consider the problem of embedding the checks $\{Z_1Z_2, X_1X_2\}$ that act on a pair of qubits located arbitrarily far apart in 2D.
\begin{equation}
    \vcenter{\hbox{\includegraphics[page=37, scale=.5]{diagrams-wire-codes.pdf}}}
\end{equation}
First we introduce appropriate Pauli branches via the degree-reduction step.
\begin{equation}
    \vcenter{\hbox{\includegraphics[page=38, scale=.5]{diagrams-wire-codes.pdf}}}
\end{equation}
Next, we map the $Z_1Z_2$ check to a $Z$-branch connecting the copy qubit belonging to the $Z$-branch of qubit $1$, to the copy qubit belonging to the $Z$-branch of qubit $2$.
\begin{equation}
    \vcenter{\hbox{\includegraphics[page=39, scale=.5]{diagrams-wire-codes.pdf}}}
\end{equation}
We perform a similar procedure for the $X_1X_2$ check.
\begin{equation}
    \vcenter{\hbox{\includegraphics[page=31, scale=.5]{diagrams-wire-codes.pdf}}}
\end{equation}

This weight reduction procedure is formalized by Algorithm~\ref{alg:check-embedding}.
\begin{algorithm}[H]
\caption{EmbedCheck}
\label{alg:check-embedding}
\begin{algorithmic}
    \Require check $s \in [1, \dots, m]$ of $\oldcodespace$
        \State Let $(q_1, \dots, q_{\omega_s})$ be the qubits in the support of $s$. 
        
        \ForEach{$q' \in (q_1, \dots, q_{\omega_s})$}
            \State Let $P_{q'}$ be the Pauli element of $s$ supported on $q'$
            \State Let $r_{q'}$ be a copy qubit from the $P_{q'}$-branch of $q'$, such that $r_{q'}$ has degree less than $3$.
            \State Add a $Z$-branch from $(r_{q'},0)$ to  $(r_{q'},s)$
        \EndFor

        \ForEach{$i \in (1, \dots, \omega_s -1)$}
            \State Add a $Z$-branch from $(r_{q_i},s)$ to $(r_{q_{i+1}},s)$
        \EndFor
\end{algorithmic}
\end{algorithm}

\paragraph{Summary of the 2D construction.}
Having introduced the elementary steps of the mapping above, the overall construction is concisely described by Algorithm \ref{alg:2d-wire-alg}. 

\begin{algorithm}[H]
\caption{2D wire code construction}
\label{alg:2d-wire-alg}
\begin{algorithmic}
    \Require A code $\mathcal{C}$, with parameters $[[n,k,d]]$, 
    $n_C$ stabilizer checks, maximum check weight $\omega$, and maximum degree~$\delta$.
    \Ensure A subsystem code that is local in 2D with parameters $[[\Theta(\delta n^2),k,O(\frac{2}{w}d)]]$.
    \ForEach {qubit $q \in [1, \dots, n]$ of $\oldcodespace$}
        \State Run DegreeReduce on $q$ (Algorithm \ref{alg:degree-reduction})
    \EndFor
    \ForEach {check $s \in [1, \dots, n_C]$ of $\oldcodespace$}
        \State Run EmbedCheck on $s$ (Algorithm \ref{alg:check-embedding})
    \EndFor  
\end{algorithmic}
\end{algorithm}

We now state Theorem~\ref{thm:2d-wire-codes} which establishes the properties of the two dimensional wire codes produced by Algorithm~\ref{alg:2d-wire-alg}.
\begin{theorem}
\label{thm:2d-wire-codes}
    Given an LDPC code $\oldcodespace$, with parameters $$[[n_\oldcode,k_\oldcode,d_\oldcode]] , $$ then the Algorithm \ref{alg:2d-wire-alg} can be used to obtain a subsystem code $\newcodespace$ with the following properties:
    \begin{enumerate}
        \item $\nnew \leq O(\nold^2 )$ 
        \item $k_\newcode = k_\oldcode $
        \item $d_\newcode \geq\Omega(d_\oldcode)$
        \item $\newcodespace$ is local in 2D
        \item The weight and degree of $\newcodespace$ is at most three
    \end{enumerate}
\end{theorem}

\begin{proof}
The number of qubits and the locality of $\newcodespace$ are by construction from Algorithm \ref{alg:2d-wire-alg}. The values of $k$ and $d$ follow from Corollary \ref{cor:k}, and Lemma \ref{lem:d}.
\end{proof}

\subsection{Wire codes in three or more dimensions}
We now move on to describe how the wire code construction works in 3D. 
In the previous section, the data qubits were laid out in 1D on the bottom row and the remaining dimension was used to connect them to the checks which were each assigned to a different row. 
For the 3D case, we follow a similar blueprint where the data qubits are laid out in 2D on a $\sqrt{n} \times \sqrt{n}$ square, while a cube of height $\sqrt{n}$ is used to accommodate connections from the qubits, which are placed on the bottom layer, to the checks, which are placed on the top layer. 

Despite these similarities, the 3D case presents a notable obstacle. From a quick calculation it can be shown that if the dimension $k$ and distance $d$ of the output code match that of the input code, then the height of the cube \emph{has} to be at most $O(\sqrt{n})$, if we are to saturate the bounds from Ref.~\cite{bravyi2011subsystem}. 
This leaves us with only $O(\sqrt{n})$ layers to work with, which implies that unlike the 2D case we have to ``pack in'' more than one check per layer.

This issue can be summarised as ``how best to `condense' long range connections between qubits and checks into a finite volume of space". There exists a very natural -- though non-trivial -- solution to this problem. Given a cube of size $\sqrt{n} \times \sqrt{n} \times \sqrt{n}$ it is possible to find edge-disjoint paths from any set of $l \leq c \cdot n$ qubits at the bottom layer to any set of $l$ qubits at the top layer, with $c$ some constant strictly smaller than $1$ \cite{thompson1977sorting}. 
For our purposes, each of these paths can be understood as a coupling -- from a data qubit to a check it is involved with in $\oldcodespace$ -- implemented by a branch of ancilla qubits along that path, together with some gauge checks in $\newcodespace$. 
In most cases the total number of qubit-check pairs induced by $\oldcodespace$ is larger than any $l$ we can choose. 
Given the guarantee that $\oldcodespace$ is sparse, i.e. $\omega,  \delta \in O(1)$, we have that the number of qubit-check pairs is bounded by $\delta \cdot n = O(n)$. 
Thus, by repeating the process of finding edge-disjoint paths for sets of appropriately chosen qubits and checks, we manage to find paths for every qubit-check we need. 
As the process is repeated $O(\frac{\delta n}{c \cdot n}) = O(1)$ times, each edge is used at most $O(1)$ times and we preserve a constant density of qubits. 

In order to apply the intuition described above, we state several lemmas before describing the general $D$-dimensional case formally. We start with the statement that we can find many edge-disjoint paths for a sufficiently thick hypercube in $D$ dimensions.
\begin{lemma}[\cite{thompson1977sorting}] 
\label{lem:Ddimrouting}Let $D\geq 2$ be an integer, and $c_D$ a universal constant. Consider the hypercube $[m]^{D-1}\times [c_Dm]$ living in $D$ Euclidean dimensions. Then for any set of pairs $\{(q_i, s_i)\}_{i=1}^{i=l}$ where $q_i$ are $l$ distinct elements of $[m]^{D-1}\times \{0\}$ and $s_i$ are $l$ distinct elements of $[m]^{D-1}\times \{c_Dm\}$ there exist $l$ edge-disjoint paths in $[m]^{D-1}\times [c_Dm]$ from each $q_i$ to its associated $s_i$.
\end{lemma}

\begin{figure}[H]
    \centering
    \includegraphics[page=41, scale=.7]{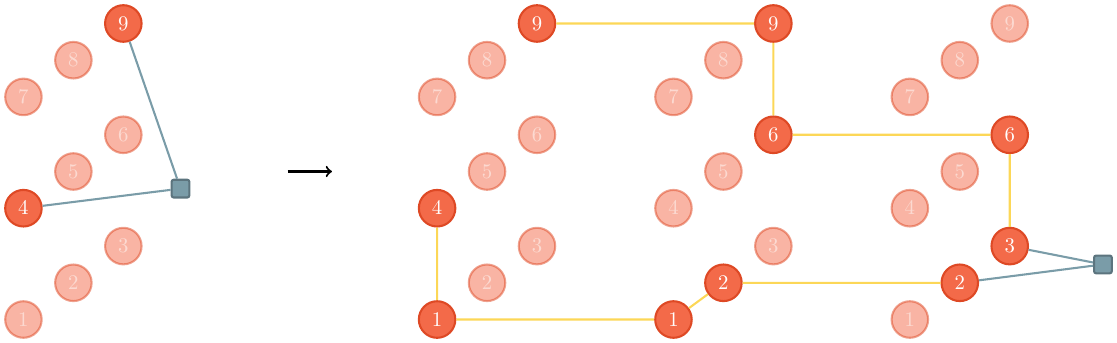}
    \caption{Illustration of how a long-range check in $D$ dimensions can be made local via routing in $D+1$ dimensions using Lemma~\ref{lem:Ddimrouting}.}
\label{fig:paths}
\end{figure}

We move on to ensure that the path-finding process only needs to be repeated a constant number of times, ensuring a constant qubit density in the resulting wire code.
\begin{lemma}
\label{lem:colourclasses}Given a code $\oldcodespace$ with maximum weight $\omega$ and maximum degree $\delta$, then the set of pairs $\{(q,s)\}_{q,s}$ induced by $\oldcodespace$ such that the qubit $q$ is in the support of the check $s$ can be partitioned into at most $\omega \cdot \delta + 1$ sets $\{\chi_j\}_j$ such that no two pairs in a given $\chi_j$ share a qubit of a check. 
\end{lemma}
\begin{proof}
    Define a graph $G = (V,E)$ that contains a vertex for every pair $(q,s)$ when $q$ is in the support of $s$, and two vertices $(q,s)$ and $(q', s')$ share an edge whenever $q = q'$ or $s = s'$. 
    Because each qubit is connected to at most $\omega$ checks and each check is connected to at most $\delta$ qubits, the degree of the graph is at most $\omega \cdot  \delta$ and thus $V$ can be partitioned into $(\omega \cdot  \delta+1)$ disjoint sets $V = \bigsqcup_j \chi_j$ such that no two vertices in a given $\chi_j$ share an edge, from a standard greedy graph colouring argument.
\end{proof}

Finally, reducing the degree of the data qubits can be done in a way similar to the 2D case. 
Indeed, for any target dimension $D \geq 3$ it is possible to embed the Pauli branches into a disc of dimension $(D-1)$ and of radius $O(\delta)$.
We go on to describe the general $D$-dimensional construction in Algorithm~\ref{alg:high-d-wire-codes-alg} below.

\begin{algorithm}[H]
\caption{$D$-dimensional Wire code construction}
\label{alg:high-d-wire-codes-alg}
\begin{algorithmic}
    \Require An LDPC code $\mathcal{C}$, with parameters $[[n,k,d]]$, $n_C$ stabilizer checks, 
    maximum check weight $\omega$, and maximum degree $ \delta$.
    \Ensure A subsystem code that is local in $D$-dimensions with parameters $[[\Theta(n\cdot c_D n^{1/(D-1)}),k,O(\frac{1}{w}d)]]$.
    \ForEach {qubit $q \in [1, \dots, n]$}
    \State Place $X$-, $Y$-, and $Z$-branches within a $D-1$-dimensional disc at an unoccupied lattice site of $[m]^{D-1}\times \{0\}$, with $m = \lceil n^{1/(D-1)}\rceil$
    \EndFor 
    \ForEach {check $s \in [1, \dots, n_C]$ } 
    \State Label a lattice site of $[m]^{D-1}\times \{c_D m\}$ corresponding to $s$. 
    \EndFor 
    \ForEach {$\chi_j$ given by Lemma \ref{lem:colourclasses}}
        \ForEach { $(q, s) \in \chi_j $}
            \State Lemma \ref{lem:Ddimrouting} guarantees the existence of a path $(a_1, a_2, \dots , s)$ from a $P$-branch of $q$ to the site labelled~$s$.
            \State Place a qubit at the location of every $a_t$
            \State Add $X_{a_t}$ to the gauge checks for every $a_t$
            \State Add $Z_{a_t}Z_{a_{t+1}}$ to the gauge checks for every pair $a_t a_{t+1}$            
        \EndFor 
        \State \multiline{Note that this $\textbf{for}$ loop might add qubits at the same coordinates as another iteration from a different class $\chi_{j'}$. In that case the qubits are simply stacked onto each other, and the gauge checks added by the current $\textbf{for}$ loop only apply to the qubits it introduced. Since the number of classes $\chi_j$ is upper bounded by $\omega \cdot \delta = O(1)$, the number of qubits can be bounded by a constant.}
    \EndFor
    \ForEach {check $s \in [1, \dots, n_C]$}
        \State \multiline{Consider the set $\{q_i\}_{i=1}^{i=\omega_s}$ of qubits in the support of $s$. In the previous loop, each of these induced a path from a $P$-branch associated to $q_i$ to $s$. Consider the final ancillary qubit that was added, $a_{T_i}$. Then we add a $Z$-branch through $(a_{T_1},a_{T_2} \dots a_{T_{\omega_s}})$.}
    \EndFor 
\end{algorithmic}
\end{algorithm}

We now state Theorem~\ref{thm:wire-codes} which establishes the properties of the $D$-dimensional wire codes produced by Algorithms~\ref{alg:2d-wire-alg} and~\ref{alg:high-d-wire-codes-alg}. 
\begin{theorem}
\label{thm:wire-codes}
    Given an LDPC code $\oldcodespace$, with parameters $$[[n_\oldcode,k_\oldcode,d_\oldcode]] , $$ and a target dimension $D \geq 3$ then the Algorithms  \ref{alg:high-d-wire-codes-alg} can be used to obtain a subsystem code $\newcodespace$ with the following properties:
    \begin{enumerate}
        \item $\nnew \leq O( c_D \nold^{D/(D-1)} )$ for some constant $c_D$
        \item $k_\newcode = k_\oldcode $
        \item $d_\newcode \geq\Omega(d_\oldcode)$
        \item $\newcodespace$ is local in $D$ dimensions
        \item The weight and degree of $\newcodespace$ is three
    \end{enumerate}
\end{theorem}

\begin{proof}
The number of qubits and the locality of $\newcodespace$ are by construction from Algorithm \ref{alg:high-d-wire-codes-alg}. The values of $k$ and $d$ follow from Corollary \ref{cor:k}, and Lemma \ref{lem:d}.
\end{proof}

\begin{corollary}
    For any $D\geq 2$, there exists an infinite sequence $\{\mathcal{C}_n\}_n$ of subsystem codes with parameters $[[n,k,d]]$ such that
    \begin{enumerate}
        \item $k = \Omega(n^{(D-1)/D}) $
        \item $d = \Omega(n^{(D-1)/D}) $
        \item $\{\mathcal{C}_n\}_n$ is local in $D$ dimensions
        \item The weight and degree of $\newcodespace$ is at most three
    \end{enumerate}
\end{corollary}
\begin{proof}
    It suffices to apply Theorem~\ref{thm:wire-codes} to a construction of good quantum LDPC codes~\cite{panteleev2020quantum,leverrier2022quantum,Dinur2023}.
\end{proof}

\subsection{Wire codes on expander graphs}
Many architectures have structures that do not reflect the usual Euclidean space. As such our previous constructions might not be readily applicable to those. Fortunately, we can further generalise our construction to arbitrary spaces, under the condition that they are ``densely connected". 

In order to help us formalise how dense of a connection is needed, we introduce some graph theoretic language. A graph $G = (V,E)$ constitutes of a set of vertices $V = [1,\dots, n]$, and a set of edges $E = \{e_i\}_i$ of the form $e = (u,v), u,v \in V$. In our context, we say that an architecture has connectivity graph $G$ when each vertex contains $O(1)$ qubits, and multi-qubit gates can only happen between qubits that are adjacent in $G$. A widespread notion of ``how well connected a graph is" is that of \emph{graph expansion}.

\begin{definition}[Graph expansion]
    A graph $G$ is $\alpha$-expander when
    \[ \min_{S \subset V} \frac{|E(S, \overline{S})|}{\min \left ( |S|, |\overline{S}| \right )} \geq \alpha \]
    The set $E(S, \overline{S})$ represents the edges with endpoints in both $S$ and its complement $\overline{S} = V \setminus S$.
    \label{def:exp}
\end{definition}

When $\alpha$ is large, it is possible to embed a large number of \emph{arbitrary} connections between some of the qubits in $G$.

\begin{lemma}[Theorem 1.2 of \cite{chuzhoy2016routing}\footnote{This theorem was originally stated for $\alpha$-well-linked graphs, but $\alpha$-expander are also $\alpha$-well-linked~\cite{WellLinked}.}]
\label{lem:routing-expanders}
    Let $G$ be an $\alpha$-expander graph on $n$ vertices. Then there exists a set $V' \subset V$ of size $|V'| \geq c_0 \frac{\alpha n}{\text{polylog}(n)}$, for some universal constant $c_0$, such that for any set of pairs $\{(q_i, s_i)\}_i$ where $q_i, s_i \in V'$ are all distinct, there exists a path from $q_i$ to $s_i$ for every $i$, and every edge belongs to at most $15$ paths.
\end{lemma}

\begin{algorithm}[H]
\caption{Wire code construction on expanders}
\label{alg:expander-wire-codes-alg}
\begin{algorithmic}
    \Require An $\alpha$-expander $G=(V,E)$ on $n$ vertices. And a quantum LDPC code $\oldcodespace$ on $n_\oldcode$ qubits, and with $m_\oldcode$ checks such that $n_\oldcode + m_\oldcode \leq c_0 \frac{\alpha n}{\text{polylog}(n)}$.
    \Ensure A subsystem code $\newcodespace$ local on $G$ with parameters $[[O(n), k_\oldcode, d_\oldcode/\omega]]$.
    \State Since $n_\oldcode + m_\oldcode \leq c_0 \frac{\alpha n}{\text{polylog}(n)}$, Lemma \ref{lem:routing-expanders} ensures that to every qubit $q \in [n_\oldcode]$, and stabilizer $s \in [m_\oldcode]$ we can assign a different vertex $\eta(q),\eta(s) \in V'$. 
    \ForEach {qubit $q \in [1, \dots, n_\oldcode]$}
        \State Place a qubit at $\eta(q)$. We refer to this qubit as $q_\eta$.
    \EndFor 
    \ForEach{ $\chi$ from Lemma \ref{lem:colourclasses}}
        \ForEach {pair $(q,s)$ in $\chi$}
                \State Let $P$ be the Pauli operator in $s$ acting on $q$.
                \State Lemma \ref{lem:routing-expanders} guarantees the existence of a path from $\eta(q)$ to $\eta(s)$.
                \State Add a $P$-branch along the path from $q_\eta$ to $\eta(s)$, placing a qubit at every vertex of the path.
        \EndFor
    \EndFor 
    \ForEach{check $s \in [1,\dots, m_\oldcode]$}
        \State Let $r_1, \dots r_{\omega_s}$ be the qubits located at $\eta(s)$ induced by the respective $P$-paths. 
        \State Add $Z_{r_1}\dots Z_{r_{\omega_s}}$ to the checks of $\newcodespace$.
    \EndFor
\end{algorithmic}
\end{algorithm}

In Algorithm \ref{alg:expander-wire-codes-alg}, for the sake of clarity, we have made no attempt to reduce the weight and the degree of the output code to three. However, this can be remedied by applying the method described in Section \ref{sec:WeightReduction}.

\begin{theorem}
\label{thm:expander-wire-codes}
    Let $G$ be an $\alpha$-expander graph of degree $O(1)$ on $n$ vertices. There is a universal constant $c_0$ such that any quantum LDPC code $\oldcodespace$ with parameters $[[n_\oldcode,k_\oldcode,d_\oldcode]]$ and maximum weight $\omega$ yields a new subsystem code $\newcodespace$ such that 

    \begin{enumerate}
        \item $\nnew \leq O(n)$ 
        \item $k_\newcode = k_\oldcode$
        \item $d_\newcode \geq {d_\oldcode}/{\omega}$
        \item $\newcodespace$ is local on $G$ 
    \end{enumerate}

    as long as $n_\oldcode + m_\oldcode \leq c_0 \frac{\alpha n}{\text{polylog}(n)}$, where $m_\oldcode$ is the number of checks of $\oldcodespace$.
\end{theorem}
\begin{proof}
    The number of qubits and the locality of $\newcodespace$ are by construction from Algorithm \ref{alg:expander-wire-codes-alg}. The values of $k$ and $d$ follow from Corollary \ref{cor:k}, and Lemma \ref{lem:d}.
\end{proof}

\begin{corollary}
    \label{cor:expander-wire-codes}
    Given an infinite family $\{G_n\}_n$ of $\alpha$-expander graphs with degree $O(1)$, Algorithm \ref{alg:expander-wire-codes-alg} yields a family $\{\mathcal{C}_n\}_n$ of subsystem codes local on $G_n$ with parameters $[[n, \Omega\left({\alpha n}/{\text{polylog}(n)}\right), \Omega\left({\alpha n}/{\text{polylog}(n)}\right)]]$.
\end{corollary}
\begin{proof}
    It suffices to apply Theorem \ref{thm:expander-wire-codes} to the good qLDPC codes from Refs.~\cite{panteleev2020quantum, leverrier2022quantum, Dinur2023}.
\end{proof}

These parameters can be compared with the best existing bounds, see Section \ref{sec:upper-bound}.

\subsection{Wire codes for weight and degree reduction}
\label{sec:WeightReduction}

In this section we describe how wire codes can be applied to reduce the weight and degree of a stabilizer code at the cost of a modest overhead. 

The weight and degree reduction proceeds similarly to the 2D wire code construction above. 
The main difference here is that we do not need to introduce any additonal qubit overhead to realize a local embedding of the Tanner graph in a hypercubic lattice. 

The weight reduction step is performed on each check of the input code individually. 
This step can be understood as performing a trivalent resolution on each check node of the input Tanner graph, similar to Figure~\ref{fig:WeightReduction}. 
Consider a check $P_1P_2\cdots P_\omega$ on a set of qubits that have been ordered $1,2,\dots,\omega$. 
we introduce ancillary qubits labelled $2.5,3.5,\dots,{(\omega-1.5)}$. 
Next, we introduce multi-qubit gauge checks $$\{P_1P_2Z_{2.5},\,Z_{2.5}P_3Z_{3.5},\,\dots,Z_{(\omega-2.5)}P_{(\omega -2)}Z_{(\omega-1.5)},\,Z_{(\omega-1.5)}P_{(\omega-1)}P_{\omega}\}$$ along with single-qubit $X$ gauge checks $\{X_i\,|\,i=2.5,3.5,\dots,(\omega-1.5)\}$. 

The degree reduction step is performed on each qubit of the input code individually. 
This step can be understood as performing a trivalent resolution on the $X$ type edges that connect to a qubit node in the input Tanner graph, and similarly for $Y$, $Z$, see Figure~\ref{fig:degree-reduction-explanation}. 
Consider a qubit that is in the $X$-type support of checks $c^X_1,c^X_2,\dots c^X_{\delta_X}$, $Y$-type support of checks $c^Y_1,c^Y_2,\dots c^Y_{\delta_Y}$, and $Z$-type support of checks $c^Z_1,c^Z_2,\dots c^Z_{\delta_Z}$, where $\delta=\delta_X+\delta_Y+\delta_Z$. 
We introduce $X$-type copy qubits labelled $1,2,\dots,(\delta_X-1)$, 
and similarly for $Y$ and $Z$. 
Next, we introduce multi-qubit gauge checks $\{X_{0}Z_1,Z_{1}Z_{2},\,Z_{2}Z_{3},\dots,Z_{(\delta_X-2)}Z_{(\delta_X-1)}\}$, where $0$ denotes the original qubit being considered. 
We now replace each $c^X_i$ check with a gauge equivalent check $\tilde{c}^X_i=c^X_i X_0Z_i$, which moves the support on qubit 0 to qubit $i$. 
Finally, we introduce single-qubit $X$ gauge checks on the copy qubits $\{X_i\,|\,i=1,2,\dots,(\delta_X-1)\}$. 
We perform an analogous procedure for checks with $Y$- and $Z$-type support on the qubit being considered. 

When applying the steps outlined above in practice, we do not need to apply the weight reduction step to any stabilizer of weight three or less. 
Similarly, we do not need to perform any degree reduction on a qubit that has degree three or less.
Furthermore, we do not need to perform the $X$-type degree reduction step on a qubit that is in the $X$ support of zero or one checks, and similarly for $Y$ and $Z$. 

At this point we have provide a recipe to transform an arbitrary input stabilizer code into a wire code with weight and degree three. 
Here, we do not count single qubit checks in the degree, as they do not require additional connectivity to implement. 
The wire codes we have described in this section are special in the sense that they have minimal edge lengths between the check and degree nodes in the weight reduced Tanner graph. 
In this sense they have minimal overhead amongst other wire codes constructed from the same input code. 
We remark that this family of wire codes is expected to preserve fault-tolerance as the number of additional qubits and gauge checks introduced, per qubit and check of the original code, is a constant. 
We summarize their properties in Theorem~\ref{thm:WeightReduction}. 

\begin{theorem}[Wire code weight reduction]
    \label{thm:WeightReduction}
    Applying the procedure described above in this section to an $[[n,k,d]]$ stabilizer code, with max check weight $\omega$ and qubit degree $\delta$, produces an $[[O(\delta n),k,\Omega(\frac{1}{\omega}d)]]$ subsystem code, which we call a wire code, with max check weight and degree three. 
\end{theorem}
\begin{proof}
    The structure of the logical operators in the weight reduction wire codes follows the general structure in Lemma~\ref{lem:bare}. 
    From this the claimed number of encoded qubits and distance follow Corollary~\ref{cor:k} and Lemma~\ref{lem:d}, respectively. 
    Bounding the number of physical qubits simply follows by noting that there are at most $\delta n$ qubits after adding the copy qubits. 
    Then adding the check ancillary qubits includes at most one extra qubit per copy qubit. 
    The bound on the check weight and qubit degree in the wire code follow directly from the construction. 
    We note again that we do not include single qubit checks when counting the qubit degree. 
    Regarding the claim about fault-tolerance above, we remark that this procedure maps a stabilizer check of weight $\omega'$ from the input code to a stabilizer of weight at most $\omega'\omega\delta/2$ in the resulting wire code. 
    This follows from counting the additional support that can be picked up due to the intersection of the original check with additional checks as depicted in Fig.~\ref{fig:xzzx-stretched}. 
\end{proof}

In the above theorem, there is a potential reduction of the distance by the constant multiple~$\frac{1}{\omega}$. 
This is a consequence of the simple trivalent graph we have used to resolve the checks of the input code in the wire code construction. 
It is possible to make alternative choices of the graph used at this step of the construction. 
Following Refs.~\cite{bacon2015sparse,Williamson2024Gauging,Ide2024} we expect that by using sufficiently expanding graphs at this step, the distance can be preserved at the cost of higher weight checks in the resulting wire code.

\subsection{Wire codes on general graphs}
\label{sec:GraphEmbedding}

In this section we discuss mapping an input stabilizer code onto a wire code that provides a local implementation on a general graph. 

The wire code construction in the previous sections made use of interactions with flexible edge lengths, mediated locally via additional ancillary qubits and gauge checks. 
More generally, every edge in the trivalent resolution of the original Tanner graph of an input stabilizer code can be assigned a positive integer length $\ell$. 
Then $O(\ell)$ additional qubits, with single qubit $X$ gauge checks and two body $ZZ$ gauge checks, can be added to maintain a local implementation, see Fig.~\ref{fig:paths}. 

Given a $c$-to-one embedding of the resolved Tanner graph of $\oldcodespace$ into a graph $G$, we can construct a wire code with qubit density $c$ that provides a local implementation of $\oldcodespace$ on $G$. 
Here, the embedding may send edges to paths consisting of many edges, but is restricted to map at most $c$ edges or vertices to one. 
The wire code construction follows by assigning edge lengths $\ell_e$ to the Tanner graph of the weight reduction wire code for $\oldcodespace$, see section~\ref{sec:WeightReduction}, given by the lengths of edges under the graph embedding into $G$. 
We then add additional $O(\ell_e)$ qubits and gauge checks to implement each length-$\ell_e$ edge locally, see Figure~\ref{fig:WeightReductionExtension}. 
Next, the Tanner graph of this edge-extended wire code can be embedded into $G$ following the original embedding map such that the checks of the wire code are local on $G$ and the density of qubits is at most $c$ on any edge or vertex.

We remark that given a sufficiently large $c$ we can embed any stabilizer code into an arbitrary graph that has at least one vertex and edge. 
Given a constant $c$ this is not always possible, and the family of graphs that a given stabilizer code can be locally implemented on generally depends on compatibility properties such as the number of vertices and edges, and the expansion of the graph. 
For embeddings that involve many edges of diverging length, we do not expect this procedure to preserve fault tolerance without some additional procedure. 

\begin{theorem}[Wire code graph embedding]
    Given an $[[n,k,d]]$ stabilizer code $\oldcodespace$ and a $c$-to-one embedding of the trivalent resolution of its Tanner graph into a graph $G$ the wire code described above in this section provides an $[[O(\ell_{\text{max}}\delta n),k,\Omega(\frac{1}{\omega}d)]]$ subsystem code that we say implements $\oldcodespace$ with local checks on $G$. 
    Here $\ell_{\text{max}}$ is the longest edge length assigned by the graph embedding. 
\end{theorem}

\begin{proof}
    Similar to the proof of Theorem~\ref{thm:WeightReduction}, the code properties of the wire code follow from Lemmas~\ref{lem:bare}, and~\ref{lem:d}. 
    Bounding the number of physical qubits follows by noting that there are at most $2 \delta n$ nodes in the resolved Tanner graph, all of which have degree-three. 
    Hence, there are at most $3 \delta n$ edges and consequently there are at most $3 \ell_{\text{max}} \delta n$ additional edge qubits. 
\end{proof}

\section{Code Properties} 
\label{sec:CodeProperties}

In this section we analyze the code properties of wire codes, including the number of physical qubits they require, the number of logical qubits they encode, their distance, relations, and how they can be used to perform syndrome extraction. 

\paragraph{Calculating $\barelogicals$.}

In the process of breaking down the checks of $\oldcodespace$ following Algorithms~\ref{alg:2d-wire-alg},~\ref{alg:high-d-wire-codes-alg},~\ref{alg:expander-wire-codes-alg} the original system on which $\oldcodespace$ lived, which we write $\mathcal{H}_{\textsf{data}} = (\mathbb{C}^2)^{\otimes\nold}$ has been augmented by the addition of supplementary registers. 
The resulting code $\newcodespace$ now lives on $\mathcal{H}_{\textsf{data}} \otimes \mathcal{H}_{\textsf{anc}} \otimes \mathcal{H}_{\textsf{copy}}= (\mathbb{C}^2)^{\otimes\nnew}$, and the implications of extending the code onto additional qubits remains to be analyzed. 
In this section, we show that the bare logical operators of the wire codes obey a rigid structure. 
This structure is crucial in our arguments to lower bound the distance of the resulting code, see below.

We restate several definitions here for convenience. 
The input stabilizer code $\oldcodespace$ is taken to live on $\nold$ qubits, and the resulting wire code $\newcodespace$ lives on $\nnew$ qubits. 
The code $\oldcodespace$ has stabilizer group $\stabs$, while $\newcodespace$ has gauge group $\gauge$. 
With these preliminaries we can define the respective logical operators of these codes 
\begin{align}
    \logicalsold &= \normalizer{\stabs} \setminus \stabs
    \\
    \logicalsnew &= \normalizer{\gauge} \cdot \gauge \setminus \gauge \ .
\end{align}
It is conventional to define a special subset of the logical operators of a subsystem code, referred to as the $\emph{bare}$ logical operators
\begin{align}
    \barelogicals = \normalizer{\gauge} \setminus \gauge.
\end{align}
We remark that the bare logical operators commute with $\gauge$, unlike $\logicalsnew$.
A similar expression will allow us to compute how many logical qubits $\newcodespace$ encodes
\begin{align}
    k_\newcode = \log_2\dim \left( \normalizer{\gauge} / \gauge \right).
\end{align}

We proceed by introducing some notation regarding the construction. Every $data$ qubit ${q \in [\nold]}$ introduces a set of $copy$ qubits. This $data$ qubit and these $copy$ qubits support a set of gauge checks denoted $\Gamma_q$. 
The set of single-site gauge checks is denoted $\Gamma_q^{\text{single-site}}$, while the remainder are denoted $\Gamma_q^{\text{copy}}$. 
Every stabilizer generator $s$ of $\oldcodespace$ is mapped to an associated set of gauge operators $\Gamma_s \subset \gauge$.
Some of these gauge operators correspond to single-site $X$ gauge operators, that we denote by $\Gamma_s^{\text{single-site}}$, while the remainder are denoted $\Gamma_s^{\text{anc}}$. 
See Fig.~\ref{fig:gamma-notation} for an illustration.

\begin{figure}[H]
    \centering
    \includegraphics[page=40, scale=.6]{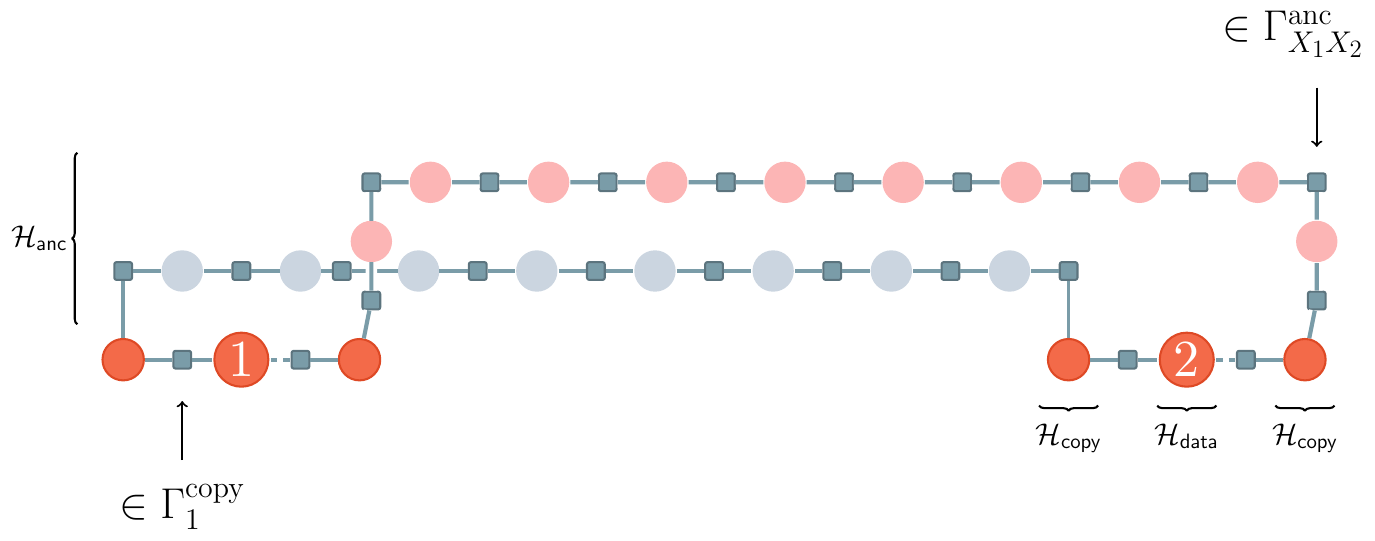}
    \caption{A simple example of a wire code for the purpose of illustrating our notation. Here, we do not explicitly depict the single-site $X$'s.}
\label{fig:gamma-notation}
\end{figure}

We say that ${\Gamma^{\text{single-site}} = \bigcup_q \Gamma_q^{\text{single-site}} \bigcup_s \Gamma_s^{\text{single-site}}}$ generates the group $\mathcal{G}_{\text{single-site}}$, while ${\Gamma^{\text{copy}} = \bigcup_q \Gamma_q^{\text{copy}}}$ generates $\mathcal{G}_{\text{copy}}$ , and $\Gamma^{\text{anc}} = \bigcup_s \Gamma_s^{\text{anc}}$ generates $\mathcal{G}_{\text{anc}}$. Naturally, $\gauge = \mathcal{G}_{\text{single-site}} \cdot \mathcal{G}_{\text{anc}} \cdot \mathcal{G}_{\text{copy}}$.
\begin{lemma}
\label{lem:bare}
    Given an input code $\oldcodespace$, with logical operators $\logicalsold = \{L_i\}_i$, then the associated wire code $\newcodespace$ satisfies
    \[
        \barelogicals = \{  L_i \otimes C_i \otimes A_i\}_i
    \]
    where $L_i$ lives on the $\data$ subsystem, while $C_i \in \mathcal{P}_{\copyr}^X$, $A_i \in \mathcal{P}_{\anc}^X$ are uniquely determined by $L_i$. 
\end{lemma}
\begin{proof}
    We start by writing $\mathcal{P}_{\nnew} = \mathcal{P}_{\data} \cdot \mathcal{P}_{\anc}^X\cdot \mathcal{P}_{\anc}^Z \cdot \mathcal{P}_{\copyr}^X\cdot \mathcal{P}_{\copyr}^Z$. 
    Because $\mathcal{G}_{\text{single-site}}$ is generated by single-site $X$'s, then for any non-identity element of $\mathcal{P}_{\anc}^Z \cdot \mathcal{P}_{\copyr}^Z$, there exists an element of $\mathcal{G}_{\text{single-site}}$ that anti commutes with it, so $\barelogicals \subset  \mathcal{P}_{\data} \cdot \mathcal{P}_{\anc}^X \cdot \mathcal{P}_{\copyr}^X$.
    
    We move on to show that $\stabs \subset \gauge$. Let $s \in \stabs$, then one can see that there exists $g_{\text{copy}} \in \mathcal{P}_{\copyr}^X$ such that
    \[\left(\prod_{g \in \Gamma_s^{\text{anc}}}  g \right) \cdot g_{\copyr} = s \otimes \idty_{\copyr} \otimes  \idty_{\anc} \]
    This establishes that for $L_{\text{wire}} \in \normalizer{\gauge}$, $\text{proj}_{\data}(L_{\text{wire}}) \in \normalizer{\stabs} $. 
    Equivalently, $L_{\text{wire}} = L\otimes C \otimes A$, for $L \in \normalizer{\stabs}, C \in \mathcal{P}_{\copyr}^X $, $A \in \mathcal{P}_{\anc}^X$. Due to the repetition-code-like nature of $\Gamma_{\text{copy}}$, $C$ can be seen to be unique. Similarly for $A$.

    Next, we need to verify that whenever $L\in \stabs$, then $L\otimes C \otimes A \in \gauge$. This is readily verified by the fact that $C \otimes A \in \mathcal{P}_{\anc}^X \cdot \mathcal{P}_{\copyr}^X \subset \langle \Gamma^{\text{single-site}} \rangle \subset \gauge $ , and 
    \[\left(\prod_{g \in \Gamma_s^{\text{anc}}}  g \right) \cdot g_{\copyr} = s \otimes \idty_{\copyr} \otimes  \idty_{\anc}  \in \gauge . \]
    Furthermore, one can verify that for $L\in \stabs$ we have $L\otimes C \otimes A \in \mathcal{S}_{\text{wire}}$ as it commutes with all gauge operators in $\gauge$. 

   Finally, we observe that for any $i,j$, we have that $L_i \otimes C_i \otimes A_i$ and $L_j \otimes C_j \otimes A_j$ have the same commutation relation as $L_i$ and $L_j$. 
   This specifically guarantees that when $L \in \logicalsold$ is non-trivial, the associated $L\otimes A \otimes C$ is non-trivial.
\end{proof}
\begin{corollary}
    \label{cor:k}
    $k_\newcode = k_\oldcode$
\end{corollary}
\begin{proof}
    From Lemma \ref{lem:bare}, we have that $\logicalsold \cong \logicalsnew$.
\end{proof}
\begin{lemma}
    \label{lem:d}
    $d_{\newcode} \geq \frac{1}{\omega} d_\oldcode $ 
\end{lemma}
\begin{proof}
    We have that 
    \[d_\text{wire} = \min_{L_\text{wire} \in \barelogicals\cdot\gauge} |L_\text{wire}| \]
    Lemma \ref{lem:bare} allows us to decompose any element $\tilde{L} \in \logicalsnew$ as $\tilde{L}=(L \otimes C \otimes A) g_{\text{single-site}} g_{\text{copy}} g_{\text{anc}}$.    
    We use this decomposition to compute the distance. 
    Indeed, $L g_{\text{copy}} g_{\text{anc}} \in \mathcal{P}_{\textsf{data}} \cdot \mathcal{P}_{\textsf{anc}}^Z \cdot \mathcal{P}_{\textsf{copy}}^Z$ while $(C \otimes A) g_{\text{single-site}} \in \mathcal{P}_{\textsf{anc}}^X \cdot \mathcal{P}_{\textsf{copy}}^X$. Hence 
\[
|\tilde{L}| = |L g_{\text{copy}} g_{\text{anc}}|+|(C \otimes A) g_{\text{single-site}} |
\]
    Because $\mathcal{G}_{\text{single-site}} =  \mathcal{P}_{\textsf{anc}}^X \cdot \mathcal{P}_{\textsf{copy}}^X$, it is always possible to find $g_{\text{single-site}}$ such that ${(C \otimes A) g_{\text{single-site}} = \idty}$. 
    The calculation of the distance $d_\newcode$ is then simplified
    \[
    d_\newcode = \min_{L\in \logicalsold, \ \ g_{\text{anc}} \in \mathcal{G}_{\text{anc}}, \ \ g_{\text{copy}} \in \mathcal{G}_{\text{copy}}} |L \ g_{\text{anc}} \ g_{\text{copy}}|
\]
    We  write $g_{\text{anc}} = \prod_{g \in T} g$, where $T$ is a subset of the generators of $\mathcal{G}_{\text{anc}}$. 
    It is natural to partition $T$ depending on what stabilizers its members are induced by
\[
T = \bigcup_s T_s, \ T_s = T \cap \Gamma_s^{\text{anc}} 
\]
    Using this notation, we can rewrite $g_{\text{anc}}$,
    \[
    g_{\text{anc}} = \prod_{s_1}^{s_m} g_{\text{anc}}(s), \ g_{\text{anc}}(s) = \prod_{g \in T_s} g.
    \]

    We now remark upon the structure of $\Gamma_{\text{anc}}$. 
    Consider a stabilizer $s$ of the form ${P_1 \otimes P_2 \otimes \dots \otimes P_{\omega}},$ $P_i \in \{X,Y,Z\} $, then  $\prod_{g \in \Gamma^{\text{anc}}_s}g = Z_1 \otimes Z_2 \otimes \dots \otimes Z_{\omega} \equiv \eta_s$. 
    It is crucial to our argument that $T_s$ follows the same behaviour. 
    To keep track of the relevant data, we write $s \in Q_1$ when the restriction of $\prod_{g \in T_s} g $ to the $\copyr$ register is equal to $\eta_s$. Otherwise, we write $s \in Q_2$.

    Let $s_1 \in Q_1$, then $|L \ g_{\text{anc}} \ g_{\text{copy}}| =  |L \ \left( \prod_{s_1}^{s_m} g_{\text{anc}}(s) \right )\ g_{\text{copy}}| \geq |L \left( \prod_{s_2}^{s_m} g_{\text{anc}}(s) \right) \eta_{s_1} g_{\text{copy}}|$. 
    Here, $\eta_{s_1}$ is supported on some $\copyr$ qubits. The next step is to address to what extent the \emph{copy} gauge checks can reduce the weight of the remaining operator. 
    Without-loss-of-generality consider a $P$-branch of a qubit $q$. The product of $\eta_s$ operators for different stabilizers $s$ supported on that branch has a certain parity. When that parity is even it corresponds to an even number of stabilizers $s$ acting on $q$, and the product of the $\eta_s$ operators can be removed by picking the appropriate $g_{\text{copy}}$. Similarly, the parity is only odd when there is an odd number of stabilizers $s$ acting on $q$, and the weight has to be at least one. 
    So for example $|L \otimes \idty_{\text{anc}} \otimes \eta_{s_1} \eta_{s_2} \dots  g_{\text{copy}}| \geq |L s_1 s_2 \dots | \geq d_{\oldcode}$.

    In the second case $g_{\text{anc}} \in Q_2$, then $\prod_{g \in T_s} g$ restricted to the $\textsf{copy}$ subsystem has weight at most $|s|-1$,  
    and has weight at least $1$ on the $\textsf{anc}$ subsystem. This implies that $\prod_{g \in T_s} g$ can remove at most $|s|-1$ elements on the $\textsf{copy}$ subsystem \emph{but} at the cost of inducing an operator of weight at least $1$ on its \textsf{anc} subsystem. 

    Combining these properties allows for the following inequality:
    \begin{align*}
         |L \ g_{\text{anc}} \ g_{\text{copy}}|  &= \left|L \left(\prod_{s \in Q_1} \prod_{g \in T_s} g  \right) \left (\prod_{s \in Q_2} \prod_{g' \in T_s} g' \right) g_{\text{copy}}\right| \\
         &\geq \frac{1}{\omega - 1} \left|L \left(\prod_{s \in Q_1} \prod_{g \in T_s} g  \right) g_{\text{copy}}\right| \\
         & \geq \frac{1}{\omega - 1} \left|L \left(\prod_{s \in Q_1} \eta_s  \right) g_{\text{copy}}\right| \\
         & \geq  \frac{1}{\omega - 1} \left|L \left(\prod_{s \in Q_1} s  \right) \right| \\
         & \geq \frac{1}{\omega - 1}  d_\oldcode
    \end{align*}
\end{proof}

\paragraph{Stabilizer check relations.}
A relation $r$ in the input code is a collection of stabilizer checks $\{s_i\}_{i\in r}$ whose product is the identity, i.e.~$\prod_{i \in r} s_i= \mathds{1}$. 
Each such relation gives rise to a relation in the associated wire code. 
That is, the collection of checks in the wire code $\Gamma_r$ given by the image of the checks $\{s_i\}_{i\in r}$ also form a relation as their product is the identity. 
Conversely, any relation in the wire code is the image of a relation from the input code. 

\paragraph{Syndrome extraction.}
The wire codes define a simple syndrome extraction procedure to measure the stabilizers of the input code. 
This procedure can be viewed as a weight reduction of the standard syndrome extraction procedure for the input code. 
\begin{itemize}
    \item 
We partition the checks of the input code into non-overlapping sets which define a measurement schedule. 
    \item 
All ancillary and copy qubits are initialized in the $\ket{+}$ state. 
    \item 
The procedure consists of two steps for each partition, which then repeat. 
\begin{itemize}
    \item The first step consists of measuring the $ZZ$, $ZZZ$, and $ZP$, gauge operators between the ancillary qubits, copy qubits and original qubits associated to the checks in the current partition. 
    This step has circuit depth 3 due to the three body checks, and maximum qubit degree 3. 
    \item The second step consists of simply measuring single-qubit $X$ gauge operators on the ancillary qubits and copy qubits associated to the checks in the current partition. 
\end{itemize}
    \item 
The procedure then moves on to the next partition. 
\end{itemize}
The above syndrome extraction procedure defines a floquet code that uses measurements of weight three and less to extract the syndromes of an arbitrary stabilizer code. 
As it is based on the wire code construction, these floquet codes can be adapted to general connectivity constraints, at the cost of qubit overhead. 
It is possible that the efficiency of the above syndrome extraction schedule can be improved based on more efficient syndrome extraction circuits for the input code and by using the degree reduction ancilla qubits to improve parallelism. 

The above procedure is very similar to Shor-style syndrome extraction~\cite{shor1996fault}. 
The fault-tolerance of the input syndrome extraction circuit can be affected by the additional qubits. 
Generically, these additional qubits are expected to lower the threshold. 
Below, we argue that if a constant number of ancilla qubits are introduced for each qubit and check of the original code, the threshold should only be reduced by a constant multiplicative factor and hence the existence of a nonzero threshold is preserved. 
If a growing number of additional qubits are used for the purposes of embedding a highly connected code into a graph with low connectivity, we expect the stabilizers to become nonlocal and there to be no threshold without the application of additional techniques. 
See Refs.~\cite{Li2018,Gidney2023} for related work in this direction. 

\paragraph{Preservation of a fault-tolerance threshold.} We now argue that wire codes derived from qLDPC codes, that only require the addition of a constant number of ancilla qubits for each check and qubit in the original code, preserve the existence of a fault-tolerance threshold when the above syndrome extraction scheme is used. 
In particular, this condition is satisfied by wire codes that directly implement weight reduction of a qLDPC code with minimal overhead. 

We consider rounds of syndrome measurement under a phenomenological noise model that consists of local Pauli and measurement faults~\cite{dennis2002topological}. 
We assume that the input code has a fault-tolerant syndrome extraction procedure which organizes checks into a constant number of groups, each of which contains non-overlapping checks which can be simultaneously measured. 
We further assume that this syndrome extraction procedure leads to a finite fault-tolerant threshold against local Pauli and measurement noise for the family of input qLDPC codes considered. 

In the syndrome extraction procedure for the associated wire codes, there are a number of additional fault locations due to the addition of ancilla qubits and the weight-reduced syndrome checks.  
We now argue that each of these faults is captured by some equivalent local fault in the original syndrome extraction procedure, and that there are a constant number of new faults contributing to any original fault. 
Hence, the wire code construction leads to an equivalent detector error model~\cite{McEwen2023,Delfosse2023} in the fault-tolerant syndrome extraction setting, but with a different rate of local errors. Since the boost to the error rate is upper bounded by a constant, the existence of a finite fault-tolerance threshold in the original code implies the existence of one for the wire code. 

In the syndrome extraction procedure for the wire code, the weight reduced checks are measured following the syndrome extraction schedule of the original code. 
This leads to a number of new fault locations that lead to effectively higher error rates. 
\begin{itemize}
    \item
First, to schedule the preparation and measurement of the ancilla qubits in the wire code can require the addition of two time steps between layers of the original syndrome measurements. These extra time steps increase the effective rate of Pauli errors per original time step by a constant factor. 
    \item 
Second, additional Pauli and measurement faults that occur on the ancilla system and gauge checks can also increase the effective error rate of the original code. 
These errors come in two types. 
\end{itemize}
The first type of errors are equivalent to local Pauli errors on the original code. 
This includes $\ket{+}$ state preparation, $X$ measurement, and Pauli $Z$ errors on the ancilla qubits of the wire code. 
All of these error types are equivalent to Pauli $Z$ errors on the ancilla qubits between their initialization and read-out in the $X$ basis. 
Any product of such Pauli $Z$ errors on the ancilla qubits associated to an original check is equivalent to a product of Pauli errors in the support of that check up to multiplication with the $ZZ$, $ZZZ$, and $ZP$ operators that are measured to infer that check. 
For an original code that is qLDPC these equivalent Pauli errors are constant weight and local on the original Tanner graph.

The second type of errors are equivalent to local measurement errors on the checks of the original code. 
This includes measurement errors of the $ZZ$, $ZZZ$, and $ZP$ gauge checks, and Pauli $X$ errors that occur between the measurement of pairs of these gauge checks. 
For wire codes that associate a constant number of ancilla qubits to each check and qubit of the original code, e.g.~from the direct weight reduction of a qLDPC code, the above syndrome extraction schedule leads to an equivalent set of detectors~\cite{McEwen2023} (formed by consecutive measurements of the original checks) undergoing a local error model that has a constant increase in the rate of additional local Pauli and measurement errors when compared to the syndrome extraction procedure of the original code. 

The above notion of equivalence between detector error models under weight reduction is similar to definitions of fault-tolerant equivalence that have been formalised in recent works Ref.~\cite{Rodatz2025}.

\section{Discussion} 
\label{sec:Discussion}

In this work we introduced wire codes. 
The wire code construction is sufficiently flexible that we were able to use it to derive a number of results. 
This includes a general weight reduction procedure for stabilizer codes that produces a related subsystem code with weight and degree three. 
By implementing long range edges via sequences of short range gauge checks, we were able to combine wire codes with graph embeddings to map stabilizer codes onto local implementations on general graphs. 
In particular, applying this procedure to map good qLDPC codes to hypercubic lattices allowed us to produce local subsystem codes that have optimal scaling code parameters in any spatial dimension, saturating the bounds from Ref.~\cite{bravyi2011subsystem}. 
Similarly, applying this procedure to map good qLDPC codes to families of expander graphs resulted in local subsystem codes with code parameters that depend on the degree of expansion. 

We anticipate that our results can be generalized to stabilizer codes based on prime dimensional qudits via a suitable modification to the qudit degree reduction step in Eq.~\eqref{eq:DegRed} and Fig.~\ref{fig:degree-reduction-explanation} that we now sketch. Rather than relying on a decomposition into $X$, $Y$, and $Z$ Pauli operators, as we have done for qubit codes, one can express an arbitrary qudit Pauli operator as a product of powers of clock and shift matrices. Using this decomposition, the degree reduction of a qudit requires two rows of ancilla qubits, one for the clock matrices and one for the shift matrices. Consider a qudit $q$ that is acted upon by a product of clock and shift matrices in a stabilizer check $s$, the weight reduced version of $s$ must act on an ancilla qudit in both the clock and shift row of ancilla qubits connected to $q$ to effectively implement the desired qudit Pauli operator on $q$.  

The results in this work open a route to map codes that have highly connected Tanner graphs and large check weight, including good qLDPC codes, to graphs with more restricted connectivity. 
The flexibility of our construction provides a useful tool to bridge the gap between hardware constraints and the requirements to implement highly efficient good qLDPC codes.

\subsection*{Acknowledgements}
    After completing this work, we learned of a forthcoming work that shares some similarity to our weight-reduced syndrome extraction procedure~\cite{JulioPreprint,Rodatz2024}. 
    NB is very grateful to Harriet Apel for visiting Sydney in February 2023, as her presence sparked the idea for this project. 
    Our results also would not be the same without Rachit Nimavat's arcane knowledge of graph theory.
    DJW thanks Michael Vasmer for useful discussions and explaining Ref.~\cite{Sabo2024} at YITP. 
    Part of this work was done while DJW was visiting the Simons Institute for the Theory of Computing. 
    This work was completed while DJW was at IBM Quantum, IBM Almaden Research Center, San Jose, CA 95120, USA. 
    DJW was supported in part by the Australian Research Council Discovery Early Career Research Award (DE220100625).

\bibliographystyle{quantum}
\bibliography{references.bib}

\appendix

\section{Upper bound on subsystem codes implemented on expander graphs}
\label{sec:upper-bound}

In this appendix we establish a bound on the code parameters of any family of subsystem codes that are implemented with local checks on a family of graphs with bounded expansion. 
Our starting point is the following geometric lemma from Ref.~\cite{baspin2023combinatorial}.
\begin{theorem}[Theorem 13 \cite{baspin2023combinatorial}]
\label{thm:comb-partitioning}
    Let $G = (V, E)$ be a graph on $n$ vertices for which there is no $U \subset V , |U | \geq m$ such that $G[U]$ is $\epsilon$-expander, then G can be partitioned into $\{A_i\}_i$ such that
    \begin{enumerate}
        \item $|A_i| < 3m$
        \item there at most $c \cdot \log(n) \cdot \alpha n$ edges with endpoints in different $A_i$’s
        
    \end{enumerate}
    for some universal constant $c$.
\end{theorem}

This result states, in essence, that if a graph does not contain an expanding subgraph, then it can always be separated into small regions at small cost. We use an amplified version of this result, which considers the question of \emph{many} expanding subgraphs.

\begin{corollary}
    \label{cor:comb-partitioning}
    Let $G = (V, E)$ be a graph on $n$ vertices. Unless there is a set $\{H_j\}_j$ of disjoint subgraphs such that 
    \begin{enumerate}
        \item $|H_j| \geq m $
        \item $\sum_j |H_j| \geq M$
        \item every $H_i$ is $\alpha$-expander
    \end{enumerate}
    Then there is a set $S$ of vertices, where $G[V \setminus S]$ is a set of non-adjacent subgraphs $\{G[A_i]\}_i$ such that
    \begin{enumerate}
        \item $|A_i| < 3m$
        \item $|S| \leq M + 2c \cdot \log(n) \cdot \alpha n$
    \end{enumerate}
\end{corollary}
\begin{proof}
    It suffices to apply Theorem \ref{thm:comb-partitioning} to the subgraph $G \setminus (\bigcup_j H_j)$.
\end{proof}

The final piece we need is the observation from Ref.~\cite{bravyi2011subsystem} that disjoint subsets of qubits of size strictly less than $d$ are jointly correctable, and $k \leq n - |A|$ whenever $A$ is a set of correctable qubits.

\begin{corollary}
    \label{cor:bound-subsystem}
    Let $[[n,k,d]]$ be a subsystem code with connectivity graph $G$. Then there are subgraphs $\{H_j\}_j$ such that 
    \begin{enumerate}
        \item $|H_j| \geq d/3 $
        \item $\sum_j |H_j| \geq \frac{1}{2}k$
        \item every $H_i$ is $\frac{k}{4 c \cdot \log(n) \cdot n}$-expander
    \end{enumerate}
\end{corollary}
\begin{proof}
    Note that by plugging $m = d/3$ in to Corollary \ref{cor:comb-partitioning} we find a set $\bigcup_i A_i$ that is correctable. 
    Using the result from Ref.~\cite{bravyi2011subsystem}, we can conclude that $k \leq |S|$. Finally, we can pick $M = \frac{1}{2}k$ to obtain the lower bound on the expansion.
\end{proof}

Optimizing for the distance, Corollary \ref{cor:bound-subsystem} states that an $\alpha$-expander connectivity graph allows for at best a $[[n, \tilde{\Theta}(\alpha n), \Theta(n)]]$ code.

\end{document}